\documentclass[11pt]{article}
\usepackage{setspace}
\usepackage[margin=2.5cm]{geometry} % Sets all margins to 2.5cm
\usepackage{natbib}
\usepackage{graphicx}
\usepackage{tikz}
\usetikzlibrary{positioning}
\usepackage{amsthm}
\usepackage{amsmath}
\usepackage{amssymb}
\usepackage{mathabx}
\usepackage{multicol}
\usepackage{url}
\usepackage{multirow}
\usepackage{algorithmic} 
\usepackage{algorithm}
\usepackage{authblk}

\usepackage[inline]{enumitem}

\usepackage{comment}
\usepackage{xcolor}
\newtheorem{theorem}{Theorem}
\newtheorem{lemma}{Lemma} 
\newtheorem{corollary}{Corollary} 

\newtheorem{definition}{Definition}

\theoremstyle{remark}
\begin{document}

\title{New Complexity Results for Fair Repetitive Scheduling}

\author[1]{Moran Koren\thanks{Supported by the Israel Science Foundation (ISF), grant No. 977/24.}}
\author[2]{Michael L. Pinedo\thanks{Supported by the United States-Israel Binational Science Foundation (BSF) grant No. 2024004.}}
\author[1]{Dvir Shabtay\thanks{Supported by the United States-Israel Binational Science Foundation (BSF) grant No. 2024004.}}

\affil[1]{\small Department of Industrial Engineering and Management, Ben-Gurion~University~of~the~Negev, Beer-Sheva, Israel, \texttt{korenmor@bgu.ac.il, dvirs@bgu.ac.il}}
\affil[2]{\small Stern School of Business, New York University, 40 West 4th Street, New York, NY, 10012, USA, \texttt{mlp5@stern.nyu.edu}}
\date{}

%\author{}
%\author{Moran Koren \and Michael L. Pinedo  \and Dvir %Shabtay} 
\maketitle

\begin{abstract}
We revisit the problem of finding fair solutions to repetitive scheduling problems with a single machine. In this problem, we are given a set of $n$ clients and a planning horizon consisting of $q$ periods (days). Each day, every client submits a single job that must be processed by the machine. The objective is to construct a set of $q$ schedules, one for each day, such that the quality of service (QoS) received by each client meets a predefined threshold. The QoS measure may be any standard scheduling criterion, such as the total waiting time or total completion time of a client's jobs over the entire planning horizon. This problem has been studied in the literature, with previous works providing complexity classifications and approximation algorithms for various QoS measures. Nevertheless, several important questions remain open. In this paper, we resolve three of these questions and identify several additional directions for future research.

\end{abstract}

%\section*{Funding}
%Moran Koren gratefully acknowledges financial support from the Israel Science Foundation (ISF), grant No. 977/24. D.~Shabtay and Michael L. Pinedo  was supported by United States-Israel Binational Science Foundation grant 2024004.

\section{Introduction}
Scheduling is a fundamental decision-making process that arises in a wide range of domains, including manufacturing systems, transportation networks, healthcare services, cloud computing, and workforce management~(see, e.g.,~\cite{Pinedo2022}). The primary objective of scheduling is to allocate limited resources over time in a manner that optimizes system performance, improves resource utilization, and satisfies operational constraints. While efficiency-related criteria such as minimizing total completion time, maximizing throughput, or reducing operational costs have traditionally dominated scheduling research, the fairness of scheduling decisions has become more and more important in modern applications~(see, e.g., \cite{Ajtai,Baum, Agnetis,DBLP:journals/eor/AgnetisBNP25,Ala,DBLP:journals/corr/abs-2603-28800}).

Fair decision-making in scheduling aims to ensure that resources, opportunities, workloads, or service levels are distributed equitably among competing entities. A fair schedule not only improves stakeholder satisfaction and trust but also prevents systematic disadvantages that may arise when optimization focuses solely on efficiency. In many real-world settings, unfair scheduling decisions can lead to reduced cooperation, lower morale, resource monopolization, and long-term performance degradation (see, e.g.,~\cite{Sedighi19}).

The importance of fairness becomes even more pronounced in repetitive scheduling environments, where scheduling decisions are made repeatedly over multiple periods~(\cite{heeger2021equitable,hermelin2025fairness,heeger2025interval,Megow}). In such settings, small inequities in individual scheduling periods may accumulate over time, resulting in significant disparities among participants. Consequently, evaluating fairness solely within a single scheduling horizon is often insufficient. Instead, fairness must be considered across the entire sequence of scheduling decisions, ensuring that advantages and disadvantages are balanced over the long term.

In this paper, we study fairness in repetitive scheduling problems. We focus on the model introduced in~\cite{hermelin2025fairness}. Although that work established a broad set of results, several important questions remain open. We resolve some of these open questions and provide further insight into the complexity of fair repetitive scheduling.

The remainder of the paper is organized as follows. In Section~\ref{sec1}, we formally present the problem, and in Section~\ref{sec2} we survey the most relevant results from the literature and state the open questions we address. Our results appear in the following two sections. Finally, in the concluding section, we discuss our findings and outline directions for future research, including additional open questions that we were unable to resolve.

%More particulaWe use a setting that was first introduced by~\cite{heeger2021equitable} and then extended in~\cite{hermelin2025fairness}  

%For this reason, repetitive scheduling introduces unique challenges that require mechanisms capable of tracking historical allocations and incorporating them into future scheduling decisions. By balancing efficiency objectives with fairness considerations over time, repetitive scheduling frameworks can achieve sustainable and equitable outcomes while maintaining high system performance. As a result, fairness has emerged as a critical criterion in the design and evaluation of modern scheduling algorithms and decision-support systems.
\section{Problem Definition}
\label{sec1}
We study a class of single-machine repetitive scheduling problems that can be defined as follows. We are given a set of $n$ \emph{clients}, $[n]:=\{1,\ldots,n\}$. Each client is associated with a single job on each of $q$ different \emph{days}. Let $(i,j)$ refer to the job of client $j$ on day $i$. Let $\mathcal{J}_{d_i}=\{(i,j):j\in [n]\}$, $\mathcal{J}_{c_j}=\{(i,j):i\in [q]\}$ and $\mathcal{J}=\{(i,j): j\in [n], i\in [q]\}$ denote the set of jobs to be scheduled on day $i$, the set of jobs that belong to client $j$, and the set of all jobs, respectively. 

An instance of our problem consists of a set of job processing times, $\{p_{ij} : (i,j) \in \mathcal{J}\}$, where $p_{ij}$ denotes the \emph{processing time} of job $(i,j)$. In addition, an instance may include other job-specific parameters, such as the \emph{due date} $d_{ij}$ or \emph{release date} $r_{ij}$ of job $(i,j)\in \mathcal{J}$. Throughout this paper, we adopt the standard assumption in static scheduling that all jobs are available at time zero, i.e., $r_{ij}=0$ for all $(i,j) \in \mathcal{J}$.

On each day $i\in [q]$, the jobs in $\mathcal{J}_{d_i}$ have to be scheduled nonpreemptively on a single machine. A solution to our repetitive scheduling problem is therefore defined by a set of $q$ permutations (schedules) $\pi = (\pi_1, \pi_2, \ldots, \pi_q)$, where $\pi_i=(\pi_i(1),\pi_i(2),\ldots,\pi_i(n))$ and $\pi_i(j)$ represents the index of the job that is scheduled on the $j$'th position on day $i$, i.e., $\pi_3(2)=5$ means that the fifth job is scheduled on the second position on day three. Given $\pi_i$, the completion time of job $(i,\pi_i(j))$ is simply the total processing time of all jobs processed on day $i$ before this job plus the processing time of the job itself. Therefore, it can be computed by $C_{i,\pi_{i} (j)} = \sum_{\ell=1}^j p_{i,\pi_{i} (\ell)}.$
%A \emph{schedule} is a permutation $\sigma: [n] \to [n]$ of the jobs.
%Given a schedule~$\sigma$, the \emph{completion time}~$C_j (\sigma)$ of job~$j$ is $C_j (\sigma) := \sum_{i \in [n] : \sigma (i) \le \sigma (j)} p_i$; that is, it is the total processing times of all jobs preceding $j$ in the schedule (including $j$ itself).
%\begin{equation}
%\label{completion}
%   C_{i,\pi_{i} (j)} = \sum_{\ell=1}^j p_{i,\pi_{i} (\ell)}.
%\end{equation}

Let $f_{ij}(C_{ij})$ be the performance measure of job $(i,j)$ indicating the QoS client $j$ receives on day~$i$. We consider four classical $f_{ij}$ functions. The first, $f_{ij}(C_{ij})=C_{ij}$, is simply the completion time of job $(i,j)$. The second, $f_{ij}(C_{ij})=W_{ij}=C_{ij}-p_{ij}$, is the waiting time of job $(i,j)$. The third, $f_{ij}(C_{ij})=L_{ij}=C_{ij}-d_{ij}$, is the lateness of job $(i,j)$. The last, $f_{ij}(C_{ij})=T_{ij}=\max\{0,L_{ij}\}$, is the tardiness of job $(i,j)$. We define the QoS received by client $j$ as $F_j=\sum_{i\in [q]} f_{ij}(C_{ij})=\sum_{(i,j)\in \mathcal{J}_{c_j}} f_{ij}(C_{ij})$, where smaller values of $F_j$ imply higher QoS.

%For $j \in [n]$, we use $F_j= \sum^q_{i=1} f_{ij}(C_{ij})$ to denote the total performance measure (or QoS) of client $j$, with smaller values indicating higher QoS. 
Taking a global perspective, one may consider the problem of finding the most \emph{efficient} solution, namely, a solution $\pi$ that minimizes $\sum_{j \in [n]} F_j = \sum_{(i,j)\in \mathcal{J}} f_{ij}(C_{ij})$. We denote this problem by $1|rep|\sum_{(i,j)\in \mathcal{J}} f_{ij}(C_{ij})$, following the classical three-field notation of~\cite{Graham}. %this problem is denoted by $1|rep|\sum_{(i,j)\in \mathcal{J}}{f_{ij}(C_{ij})}$).
%Consider first the single machine problem which aims to compute the most \emph{efficient} solution $\pi$ that minimizes $\sum_{j \in [n]} F_j=\sum_{(i,j)\in \mathcal{J}}{f_{ij}(C_{ij})}$. Following the classical three-field notation of~\cite{Graham}, this problem is denoted by $1|rep|\sum_{(i,j)\in \mathcal{J}}{f_{ij}(C_{ij})}$. Finding the most efficient solution reduces to $q$ independent single-machine problems, one for each day $i\in[q]$ (\cite{hermelin2025fairness}). Therefore, and based on classical results for the single day counterpart, the $1|rep|\sum_{(i,j)\in \mathcal{J}}{f_{ij}(C_{ij})}$ is solvable in $O(qn\log n)$ time for $f_{ij}(C_{ij})\in \{C_{ij},W_{ij},L_{ij}\}$ by ordering jobs on each day $i\in[q]$ according to the classical shortest processing time (SPT) rule (\cite{Smith1956}). For $f_{ij}(C_{ij})=T_{ij}$ the single day problem is ordinary NP-hard~(\cite{Du1990}), and solvable in pseudo polynomial $O(n^4 P)$ time~(\cite{Lawler77}), where $P$ is the total processing time of all jobs. Consequently, finding an efficient solution to $1|rep|\sum_{(i,j)\in \mathcal{J}}T_{ij}$ can be done in $O(qn^4 \max_i{P}_i)$ time, where $P_i=\sum_{(i,j) \in \mathcal{J}_{d_i}} p_{ij}$ denotes the total processing time of all jobs to be processed on day $i$. 
Nevertheless, optimizing overall efficiency may lead to highly unbalanced outcomes, where some clients experience substantially poorer QoS than others. To promote fairness among clients, we adopt the classical min--max criterion. Specifically, our goal is to compute a solution $\pi$ that minimizes the maximum $F_j$ value incurred by any client. We refer to this problem as the \emph{maximal fairness} problem, which, based on the classical three-field notation, is denoted by $1|rep|\max_j \sum_i f_{ij}(C_{ij})$.
%In some cases, we consider the decision version of the problem, which asks whether there exists a $K$-fair solution that satisfies $\max_j F_j \leq K$. We denote this problem by $1|rep, \max_j \sum_i f_{ij}(C_{ij})\leq K|-$

\section{Related Work and Research Questions}
\label{sec2}
\cite{hermelin2025fairness} observed that solving the $1|rep|\sum_{(i,j)\in \mathcal{J}}{f_{ij}(C_{ij})}$ problem (i.e., finding the most efficient solution) reduces to $q$ classical and independent $1||\sum_j f_{j}(C_{j})$ problems, one for each day $i\in[q]$. Therefore, and based on classical results for the single day counterpart (see,~\cite{Smith1956}), the $1|rep|\sum_{(i,j)\in \mathcal{J}}{f_{ij}(C_{ij})}$ is solvable in $O(qn\log n)$ time for $f_{ij}(C_{ij})\in \{C_{ij},W_{ij},L_{ij}\}$ by ordering jobs on each day $i\in[q]$ according to the classical shortest processing time (SPT) rule (\cite{Smith1956}). For $f_{ij}(C_{ij})=T_{ij}$ the single day problem is ordinary NP-hard~(\cite{Du1990}), and solvable in pseudo polynomial $O(n^4 P)$ time~(\cite{Lawler77}), where $P$ is the total processing time of all jobs. Consequently, finding an efficient solution to $1|rep|\sum_{(i,j)\in \mathcal{J}}T_{ij}$ can be done in $O(qn^4 \max_i{P}_i)$ time, where $P_i=\sum_{(i,j) \in \mathcal{J}_{d_i}} p_{ij}$ denotes the total processing time of all jobs to be processed on day $i$. 

Consider now the maximal fairness problem, $1|rep|\max_j \sum_{i}f_{ij}(C_{ij})$. When $q=1$, the problem reduces to the classical (single day) min-max $1||\max_j f_{j}(C_{j})$ problem which is easy for any regular performance measure (being a non-decreasing function of $C_{ij}$), see~\cite{Lawler73}. However, in special cases, the problem can be solved using a faster algorithm. In particular, when $f_{j}(C_{j})=C_{j}$ any permutation is optimal. When $f_{j}(C_{j})=W_{j}$ any permutation with the last scheduled job having the largest processing time is optimal. When $f_{j}(C_{j})=T_{ij}$ ordering the jobs in a non-decreasing order of due dates (i.e., according to the EDD rule) is optimal, see~\cite{Jackson1955}. 

\cite{hermelin2025fairness} studied the $1|rep|\max_j \sum_{i} f_{ij}(C_{ij})$ problem for the case of $q \ge 2$ days. They considered the three performance functions $f_{ij}(C_{ij}) \in \{C_{ij}, W_{ij}, L_{ij}\}$ and proved that the problem is strongly NP-hard when $q$ is part of the input, regardless of the chosen performance function. They therefore turned their attention to the computational complexity of the problem when the number of days is fixed.

%They showed that the problem is strongly NP-hard when $q$ is arbitrary for any of the three different $f_{ij}(C_{ij})$ function. However, when $q$ 
\cite{hermelin2025fairness} proved that the special case $1|rep,q=2|\max_j \sum_i C_{ij}$, corresponding to a planning horizon of two days, can be solved in $O(n\log n)$ time. Interestingly, their algorithm is similar to the classical Johnson's algorithm for scheduling in a two-machine flow shop setting~(\cite{johnson1954optimal}), where the two days play the role of the two machines. The main difference is that the algorithm for $1|rep,q=2|\max_j \sum_i C_{ij}$ schedules the jobs on the first day using Johnson's sequencing rule, and uses the reverse order for the second day, whereas in the two-machine flow shop problem the job sequence is identical on both machines. In contrast to this positive result,~\cite{hermelin2025fairness} proved that the $1|rep|\max_j \sum_i C_{ij}$ problem is weakly NP hard for any fixed number of days which is at least four. The main gap that remains open with respect to the problem is the case of three days. This leads to our first research question.
\begin{itemize}
    \item Question 1: Can we solve the $1|rep|\max_j \sum_i C_{ij}$ problem in polynomial time when $q=3$?
\end{itemize}

\cite{hermelin2025fairness} also studied the $1|rep|\max_j \sum_i W_{ij}$ problem and showed that it is weakly NP-hard for any fixed $q \ge 4$. However, its complexity remains open for $q\in\{2,3\}$, which motivates our second research question.
\begin{itemize}
    \item Question 2: Can we solve the $1|rep|\max_j \sum_i W_{ij}$ problem in polynomial time for $q=3$? If not, can we do so for $q=2$?
\end{itemize}
We note here that the complexity landscape of the $1|rep|\max_j \sum_{i} L_{ij}$ problem with respect to the number of days is much clearer, as \cite{hermelin2025fairness} proved that this problem is weakly NP-hard already when $q=2$. Nevertheless, the proof is built upon a reduction where some of the scheduled jobs have negative lateness, and therefore the following question is still open.
\begin{itemize}
    \item Question 3: Can we solve the $1|rep|\max_j \sum_i T_{ij}$ problem in polynomial time for $q\in \{2,3\}$?
\end{itemize}

Extensions of the model studied in~\cite{hermelin2025fairness}, focusing on the design of heuristic and approximation algorithms, can be found in~\cite{shab,segev,Megow}.

\section{Hardness of $1|rep|\max_i\sum_j{C_{ij}}$ on three days}
\label{hardness}
In this section, we prove that $1|rep|\max_i\sum_j{C_{ij}}$ is NP-hard when $q=3$. The proof is based on polynomial-time reduction from the classical NP-hard Partition problem.
\begin{definition}
\label{def:par}
\textsc{Partition}: Given a set of positive integers $\mathcal{A}=\{a_1,...,a_{h}\}$ ($\sum_{a_i \in \mathcal{A}} a_i = 2B$), determine whether there exists a partition of $\mathcal{A}$ into two subsets $\mathcal{A}_1$ and $\mathcal{A}_2$ such that $\sum_{a_i \in \mathcal{A}_1} a_i=\sum_{a_i \in \mathcal{A}_2} a_i=B$.
\end{definition}

%\subsection{Hardness of $1|rep|\max_i\sum_j{C_{ij}}$ on three days}

\begin{theorem}%[\appref{thm:fewdays}]
\label{thm:fewdaysNPh}
The $1|rep,q=3|\max_i\sum_j{C_{ij}}$ problem is weakly NP-hard. 
\end{theorem}
\begin{proof}
Given an instance of \textsc{Partition}, we construct an instance for $1|rep|\max_i\sum_j{C_{ij}}$ with $n=|\mathcal{A}|+2$ clients and $q=3$ days as follows:
\begin{itemize}
    \item We create one client $x$ with $p_{1,x}=4B$, $p_{2,x}=3B$, and $p_{3,x}=0$.
    \item We create one client $y$ with $p_{1,y}=0$, and $p_{2,y}=p_{3,y}=2B$.
    \item For each $a_j\in \mathcal{A}$ we create one client $j$ with $p_{1,j}=p_{3,j}=a_j$ and $p_{2,j}=0$.
    \item We set $K=8B$ and ask whether there exists a schedule with $\max_i\sum_j{C_{ij}} \leq K$.
\end{itemize}
The reduction obviously requires only polynomial time.

\emph{Correctness.} We now show that the \textsc{Partition} instance is a yes-instance if and only if the constructed $1|rep|\max_i\sum_j{C_{ij}}$ instance is a yes-instance.

($\Rightarrow$): Assume that the given \textsc{Partition} instance is a yes-instance. Accordingly there exists a partition of $\mathcal{A}$ into $\mathcal{A}_1$ and $\mathcal{A}_2$ such that $\sum_{a_j\in \mathcal{A}_1}a_j=\sum_{a_j\in \mathcal{A}_2}a_j=B$. We then schedule the jobs as follows:
\begin{itemize}
    \item On day 1, we first schedule the jobs of clients $j$ with $a_j\in \mathcal{A}_1$ in an arbitrary order. Then we schedule the job of client $x$. Afterwards, we schedule the jobs of clients $j$ with $a_j\in \mathcal{A}_2$ in an arbitrary order.
    \item On day 2, we first schedule the job of client $x$, followed by the job of client $y$.
    \item On day 3, we first schedule the jobs of clients $j$ with $a_j\in \mathcal{A}_2$ in an arbitrary order. Then we schedule the job of client $y$. Finally, we schedule the jobs of clients $j$ with $a_j\in \mathcal{A}_1$ in an arbitrary order.
\end{itemize}

Note that we did not specify above how to schedule jobs with zero processing times, which are scheduled first on their respective days in arbitrary order.

Consider first the jobs of client $x$. On day 1, its job is completed at time $5B$, whereas on day 2 it is completed at time $3B$. Since the job of client $x$ on the last day has zero processing time, it follows that the total completion time of client $x$'s jobs is exactly $K = 8B$. By a similar argument, it is easy to verify that the total completion time of client $y$'s jobs over the three days is also exactly $8B$. Now, consider an arbitrary client~$j'$ with $a_{j'}\in A_1$. The completion time of this client's job is at most $B$ on the first day and at most $4B$ on the third day. Since the job of any client $j' \in[h]$ is completed at time zero on the second day, the total completion time of any client $j'$  with $a_{j'}\in A_1$ is at most $5B<K$. By a similar argument, the total completion time of any client $j$ with $a_j\in A_2$ is at most $7B<K$.

($\Leftarrow$): Assume that the constructed $1|rep|\max_i\sum_j{C_{ij}}$ instance is a yes-instance, i.e., there is a schedule $\pi$, with $\max_i\sum_j{C_{ij}}\leq K$. We begin by proving some properties regarding $\pi$ that will help us later to prove that we have a yes solution to the Partition instance as well.
\begin{lemma}
In schedule $\pi$, the job of client $x$ is scheduled before the job of client $y$ on day 2.
\end{lemma}
\begin{proof}
Assume toward a contradiction that in $\pi$ the job of client $y$ is scheduled before the job of client $x$ on day 2. Then the completion time of the client $x$ job on day 2 is $5B$. Combined with the processing time of $4B$ on day 1, the total completion time of client $x$ is at least $9B>K$, a contradiction.
\end{proof}
%First, we show that in $\pi$, the job of client $x$ is scheduled before the job of client $y$ on day 2. Otherwise, the completion time of client $x$'s job on day 2 would be $5B$. Together with the fact that client $x$'s job requires $4B$ processing time on day 1, this yields a total completion time of at least $9B>K$ for client $x$, a contradiction.
Now, for schedule $\pi$, let $P$ denote the set of clients whose jobs are scheduled before the job of client $x$ on day 1, and let $Q$ denote the set of clients whose jobs are scheduled after the job of client $x$ on day 1.
Moreover, let $R$ denote the set of clients whose jobs are scheduled before the job of client $y$ on day 3, and let $V$ denote the set of clients whose jobs are scheduled after the job of client $y$ on day 3.

%The total completion time of client $x$ jobs is
%$$C_x=\sum_{i=1}^3 C_{i,x}=\sum_{j \in P} p_{ij}+p_{1x}+p_{2x}=\sum_{j \in P} a_{j}+4B+3B.$$
%The fact that $C_x\leq K=8B$ then leads to, 
%$$\sum_{j \in P} a_{j}\leq B.$$
%The total completion time of client $y$ jobs is
%$$C_y=\sum_{i=1}^3 C_{i,y}=5B+\sum_{j \in R} a_{j}+2B.$$
%The fact that $C_y\leq K=8B$ then leads to, 
%$$\sum_{j \in R} a_{j}\leq B.$$
For each client $j \in Q$, define $\Pi(j)$ as the set of clients in $Q$ that are scheduled before job $(1,j)$ (including job $(1,j)$ itself). Similarly, for each client $j \in V$, define $\Sigma(j)$ as the set of clients in $V$ that are scheduled before the job $(3,j)$ (including job $(3,j)$ itself). 
\begin{lemma}
$W=Q \cap V=\emptyset$.
\end{lemma}
\begin{proof}
By contradiction, assume that $W\neq \emptyset$. Then, for any $w\in W$, we have
$$C_w=C_{1w}+C_{3w}=\sum_{j \in P} p_{1j}+p_{1x}+\Pi(w)+p_{3y}+\sum_{j \in R} p_{3j}+\Sigma(w)=$$
$$\sum_{j \in P} a_{j}+4B+\Pi(w)+\sum_{j \in R} a_{j}+2B+\Sigma(w)\leq K=8B,$$
which leads to the following inequality:
\begin{equation}
\label{first1}
\Pi(w)+\Sigma(w)\leq 2B-\sum_{j \in P} a_{j}-\sum_{j \in R}a_{j}.
\end{equation}
Using the inclusion-exclusion relation among sets, the following holds as well:
\begin{align}
\label{second2}
\sum_{j\in W} a_j & = \sum_{j\in Q\cap V} a_j=\sum_{j\in Q} a_j+\sum_{j\in V} a_j-\sum_{j\in Q\cup V} a_j \ge \nonumber \\
&  \big(2B-\sum_{j \in P} a_{j} \big)+\big(2B-\sum_{j \in R} a_{j}\}-2B=2B-\sum_{j \in P} a_{j}-\sum_{j \in R} a_{j},
\end{align}
%$$\sum_{j\in W} a_j=\sum_{j\in Q\cap V} a_j=\sum_{j\in Q} a_j+\sum_{j\in V} a_j-\sum_{j\in Q\cup V} a_j \ge \big(2B-\sum_{j \in P} a_{j} \big)+\big(2B-\sum_{j \in R} a_{j}\}-2B=2B-\sum_{j \in P} a_{j}-\sum_{j \in R} a_{j},$$
where the inequality follows from the fact that $Q\cup V \subseteq [h]$ (recall that $h$ is the number of elements in the Partition instance).
Now, let $w^* \in W$ be the element in $W$ scheduled last on day 1. Consequently, any other element in $W$ is scheduled before $w^*$ on day 1, and we have $\Pi(w^*)\ge \sum_{j \in W} a_{j}$. Plugging this into~(\ref{first1}) and using the fact that $\Sigma(w^*)\ge a_{w^*}\ge 1$,
\begin{equation}
    \sum_{j \in W} a_{j} \leq \Pi(w^*) \leq 2B-\sum_{j \in P} a_{j}-\sum_{j \in R}a_{j}-1<2B-\sum_{j \in P} a_{j}-\sum_{j \in R}a_{j},
\end{equation}
which yields a contradiction with~(\ref{second2}).
\end{proof} 

The fact that $W=Q \cap V=\emptyset$, further implies that $Q=R$ and $P=V$, i.e., any job $(1,j)$ that is scheduled before the job of client $x$ on day $1$ is scheduled after the job of client $y$ on day 3. Therefore, in $\pi$:
\begin{itemize}
\item The total completion time of the client $x$ jobs is
$$C_x=C_{1x}+C_{2x}=\sum_{j\in P} a_j+4B+3B=7B+\sum_{j\in P} a_j\leq K=8B,$$
and therefore
\begin{equation}
\label{par11}
\sum_{j\in P} a_j\le B.
\end{equation}
\item The total completion time of the client $y$ jobs is
$$C_y=C_{2y}+C_{3y}=5B+\sum_{j\in Q} a_j+2B=7B+\sum_{j\in Q} a_j =9B-\sum_{j\in P} a_j\leq K=8B,$$
and therefore 
\begin{equation}
\label{par21}
\sum_{j\in P} a_j\ge B.
\end{equation}
\end{itemize}
It follows from eqs.(\ref{par11})-(\ref{par21}) that $\sum_{j\in P} a_j = \sum_{j\in Q} a_j=B$. Defining $\mathcal{A}_1=\{a_j:j\in P\}$ and $\mathcal{A}_2=\{a_j:j\in Q\}$, we obtain a yes solution to the Partition instance.
%Consider day 2. For all clients $j$ whose jobs are scheduled earlier than the job of client $x$ on day 2, we put $a_j$ into $A_1$. For all remaining clients~$j$, we put $a_j$ into $A_2$. Assume for contradiction that $\sum_{a_j\in A_1}a_j\neq\sum_{a_j\in A_2}a_j$. Consider first the case $\sum_{a_j\in A_1}a_j>\sum_{a_j\in A_2}a_j$. Then $C_x(\mathcal{S})=7B+\sum_{a\in A_1}a>8B=K$, a contradiction. Now consider the case $\sum_{a_j\in A_1}a_j<\sum_{a_j\in A_2}a_j$. Assume for contradiction that all clients~$j$ with $a_j\in A_2$ have their job scheduled earlier than client~$y$ on day 4. Then by a similar argument as above we get that $C_y(\mathcal{S})>8B=K$, a contradiction. Hence we have that there is a client $j$ that has its job scheduled after client $x$ on day 2 and after client $y$ on day 4. However, then we have $C_j(\mathcal{S})>10B>K$, a contradiction.
\end{proof}
Note that $T_{ij}=C_{ij}$ when $d_{ij}=0$. Accordingly, the hardness result presented in this section carries to the  $1|rep|\max_i\sum_j{T_{ij}}$ with $q=3$. Nevertheless, we will show in the next section that this problem is NP-hard even when $q=2$.

\section{Hardness of $1|rep|\max_i\sum_j{W_{ij}}$ on two days}
In this section, we prove that $1|rep|\max_i\sum_j{W_{ij}}$ is NP-hard when $q=2$. Here as well, the proof is based on a polynomial-time reduction from the classical NP-hard Partition problem (see Definition~\ref{def:par}).
\begin{theorem}
\label{T2}
The $1|rep,q=2|\max_i\sum_j{W_{ij}}$ problem is weakly NP-hard, even when $p_{ij}=p_j$ for every $(i,j) \in \mathcal J$. 
\end{theorem}
\begin{proof}
Given an instance of \textsc{Partition}, we construct an instance of the $1|rep|\max_i\sum_j{W_{ij}}$ problem with $n=|A|+2$ clients and $q=2$ days as follows:
\begin{itemize}
    \item We create two clients $x$ and $y$ with processing time $5B$ on both days, i.e., $p_{i,x}=p_{i,y}=5B$ for $i\in\{1,2\}$. We call the jobs of these clients 'big' jobs. 
    \item For each $a_j\in A$ we create one client $j$ with $p_{1,j}=p_{2,j}=a_j$. We call the jobs of these clients 'small' jobs. 
    \item We set $K=8B$, and ask whether there exists a schedule with $\max_i\sum_j{W_{ij}} \leq K$.
\end{itemize}
The reduction obviously requires only polynomial time, and satisfies the condition that $p_{ij}=p_j$ for every $(i,j) \in \mathcal J$.

\emph{Correctness.} We now show that the \textsc{Partition} instance is a yes-instance if and only if the constructed $1|rep|\max_i\sum_j{W_{ij}}$ instance is a yes-instance.

($\Rightarrow$): Assume that the given \textsc{Partition} instance is a yes-instance. Accordingly there exists a partition of $\mathcal{A}$ into $\mathcal{A}_1$ and $\mathcal{A}_2$ such that $\sum_{a_j\in \mathcal{A}_1}a_j=\sum_{a_j\in \mathcal{A}_2}a_j=B$. We then schedule the jobs as follows:
\begin{itemize}
    \item On day 1, we first schedule the job of clients $j$ with $a_j\in \mathcal{A}_1$ in an arbitrary order. Then we schedule the job of client $x$. Afterwards, we schedule the jobs of clients $j$ with $a_j\in A_2$ in an arbitrary order. Finally, we schedule the job of client $y$.
    \item On day 2, we begin by scheduling the jobs of clients $j$ with $a_j\in \mathcal{A}_2$ in an arbitrary order. Then we schedule the job of client $y$. Afterwards, we schedule the jobs of clients $j$ with $a_j\in \mathcal{A}_1$ in an arbitrary order. Lastly, we schedule the job of client $x$.
\end{itemize}

Consider first the jobs of client $x$. On day 1, its job waits until time $B$ before its processing starts, whereas on day 2 it waits until time $7B$. Similarly, the waiting times of the jobs of client $y$ are $7B$ and $B$ on day 1 and day 2, respectively. Therefore, the total waiting time of either one of the two clients $x$ and $y$ is exactly $K = 8B$. Now, consider an arbitrary client~$j'$ with $a_{j'}\in \mathcal{A}_1$. The waiting time of this client's job is at most $B$ on the first day and at most $7B$ on the second day. Therefore, the total waiting time of any client $j'$  with $a_{j'}\in \mathcal{A}_1$ is at most $K=8B$. Using a similar argument, the total waiting time of any client $j$ with $a_j\in \mathcal{A}_2$ is at most $K=8B$.

($\Leftarrow$): Assume that the constructed $1|rep|\max_i\sum_j{W_{ij}}$ instance is a yes-instance, i.e., there is a schedule $\pi$, with $\max_i\sum_j{W_{ij}}\leq K$. We begin by proving some properties regarding $\pi$.
\begin{lemma}
\label{first}
In $\pi$ the job of client $x$ is scheduled before the job of client $y$ on exactly one of the two days.
\end{lemma}
\begin{proof} 
By contradiction, assume that in $\pi$ the job of client $x$ is scheduled either before (case $i$) or after (case $ii$) the job of client $y$ on both days. In the first case, the job of client $y$ waits at least $5B$ time units on each of the two days, so its total waiting time is at least $10B>K$, a contradiction. The second case is symmetric: the total waiting time of client $x$ is at least $10B>K$, again a contradiction.
\end{proof}

\begin{lemma}
\label{second}
In $\pi$ neither of the small jobs is scheduled after both of the big jobs on either of the days.
\end{lemma}
\begin{proof} 
By contradiction, assume that in $\pi$ there exists an arbitrary small job $j$ that is scheduled after the two big jobs in one of the days. As the processing time of a big job is $5B$, job $j$ will wait at least $10B$ time units on that day, a contradiction.
\end{proof}

By Lemma~\ref{first}, we may assume, without loss of generality, that in $\pi$ the job of client $x$ is scheduled before the job of client $y$ on the first day, and after it on the second day. Now, let $\alpha_1$ be the set of small clients whose jobs are scheduled before the job of client $x$ on day 1, and $\alpha_2$ be the set of small clients whose jobs are scheduled after the job of client $x$ and before the job of client $y$ on day 1. In a similar fashion, let $\beta_1$ be the set of small clients whose jobs are scheduled before the job of client $y$ on day 2, and $\beta_2$ be the set of small clients whose jobs are scheduled after the job of client $y$ and before the job of client $x$ on day 2. Note that by Lemma~\ref{second}, $\alpha_1\cup\alpha_2=\beta_1\cup\beta_2=[h]$.
\begin{lemma}
\label{third}
$\alpha_2\cap\beta_2=\emptyset$.
\end{lemma}
\begin{proof} 
By contradiction, assume that in $\pi$ $\alpha_2\cap\beta_2 \neq \emptyset$, and consider an arbitrary client $j \in W=\alpha_2\cap\beta_2$. As there is a single big job scheduled before the job of client $j$ on either one of the two days, we have $W_{1j}+W_{2j}\ge 10B>K$, a contradiction.
\end{proof}
It follows from Lemmas~\ref{second} and \ref{third} that $\beta_2=\alpha_1$ and that $\beta_1=\alpha_2$. Accordingly,
\begin{itemize}
\item The total waiting time of the client $x$ jobs in $\pi$ is
$$W_{1x}+W_{2x}=\sum_{j\in \alpha_1} a_j+2B+5B=7B+\sum_{j\in \alpha_1} a_j\leq K=8B,$$
and therefore
\begin{equation}
\label{par1}
\sum_{j\in \alpha_1} a_j\le B.
\end{equation}
\item The total waiting time of the client $y$ jobs in $\pi$ is
$$W_{1y}+W_{2y}=2B+5B+\sum_{j\in \alpha_2} a_j=7B+\sum_{j\in \alpha_2} a_j =9B-\sum_{j\in \alpha_1} a_j\leq K=8B,$$
and therefore 
\begin{equation}
\label{par2}
\sum_{j\in \alpha_1} a_j\ge B.
\end{equation}
\end{itemize}
It follows from eqs.(\ref{par1})-(\ref{par2}) that $\sum_{j\in \alpha_1} a_j = \sum_{j\in \alpha_2} a_j=B$. Defining $\mathcal{A}_1=\{a_j:j\in \alpha_1\}$ and $\mathcal{A}_2=\{a_j:j\in \alpha_2\}$, we obtain a yes solution for the Partition instance.
\end{proof}
Note that when $d_{ij}=p_{ij}$ the relation that $T_{ij}=W_{ij}$ holds. Therefore, and based on Theorem~\ref{T2}, the following corollary holds, providing a negative answer to Question 3 as well.
\begin{corollary}
    The $1|rep,q=2|\max_i\sum_j{T_{ij}}$ is weakly NP-hard, even when $p_{ij}=p_j$ for every $(i,j) \in \mathcal J$.
\end{corollary}

\section{Summary and Open Questions}
In this paper, we address three open questions regarding the complexity status of the maximal fairness problem in repetitive scheduling with respect to the number of periods (days). Our results provide a clear characterization of how the number of days affects the complexity of the fairness problem when considering the completion time, waiting time and tardiness criteria.

More specifically, when the quality of service of each client is measured by the total completion time of their jobs, the problem is polynomial-time solvable when the number of days is at most two, NP-hard when the number of days is fixed and greater than two, and strongly NP-hard for an arbitrary number of days. When the quality of service is instead measured using either the waiting time or the tardiness criterion, similar complexity results hold, with the difference that the problem becomes NP-hard already for two days.  

Several important questions regarding the problems that we consider remain open. In the following, we list those that we consider the most fundamental.

\begin{itemize}
    \item The $1|rep|\max_j\sum_i f_{ij}(C_{ij})$ problem is strongly NP-hard when $f_{ij}(C_{ij})\in\{C_{ij},W_{ij},L_{ij},T_{ij}\}$. When the number of days $q$ is fixed, the problem is known to be weakly NP-hard for any constant $q \ge 3$ under the completion-time criterion, and for any constant $q \ge 2$ under the waiting-time, lateness, and tardiness criteria. However, whether these fixed-$q$ cases are solvable in pseudo-polynomial time remains an open question.
    \item The special case in which the processing time of each job is independent of the day ($p_{ij}=p_j$ for every $(i,j) \in \mathcal J$) has numerous practical applications. For example, each client may submit the same job in every period. One such application arises in healthcare, where a patient undergoes treatment using the same machine (e.g., a radiation therapy machine) during multiple visits, with the same treatment administered in each session. Despite its practical importance, the computational complexity of this special case remains unknown for the completion time and lateness criteria.
%    \item The unit-processing-time case is solvable in polynomial time for all three performance measures considered in this paper. However, the complexity of the unit-processing-time case remains unknown when the performance measure is tardiness, defined as the maximum of zero and the lateness of a job. 
\end{itemize}
Several promising directions for future research remain. These include analyzing the price of fairness, investigating alternative notions of fairness, and extending the model to more complex machine environments (initial steps toward this goal have been presented in~\cite{Megow}). Another important direction is the development of practical heuristic algorithms capable of efficiently solving the computationally hard variants. Finally, an especially promising avenue is the study of fairness in repetitive optimization problems arising in other domains, such as matching and vehicle routing. 

%In this paper, we begin by providing a polynomial time procedure to solve the decision version of the problem, i.e., to given $K$, we ask if there exists a $K$-fair solution with $L_j=\sum_{i\in[q]} L_{ij} \leq K$ in polynomial time. Then, using a binar we use a binary to   

\bibliographystyle{apalike-ejor}
\bibliography{bib}	 
\end{document}